\documentclass[runningheads]{llncs}

\usepackage{amsfonts}
\usepackage{amsmath}
\usepackage{amssymb}
\usepackage{bm}
\usepackage{algorithmic}
\usepackage{multirow}
\usepackage{booktabs}
\usepackage{subfig}
\begin{document}
\pagestyle{headings}
\titlerunning{Robust Line Planning in case of Multiple Pools and Disruptions}
\authorrunning{A.~Bessas, S.~Kontogiannis, and C.~Zaroliagis}


\newcommand{\myvector}[1]{\ensuremath{\bm{#1}}}
\newcommand{\mv}[1]{\myvector{#1}}
\newcommand{\hvec}[1]{\ensuremath{\myvector{\hat{#1}}}}
\newcommand{\tvec}[1]{\ensuremath{\myvector{\tilde{#1}}}}
\newcommand{\bvec}[1]{\ensuremath{\myvector{\bar{#1}}}}
\newcommand{\matrixrow}[2]{\ensuremath{{\myvector{#1}_{#2,\star}}}}
\newcommand{\matrixcol}[2]{{\ensuremath{\myvector{#1}_{\star,#2}}}}
\newcommand{\mymatrix}[1]{\ensuremath{\bm{#1}}}
\newcommand{\mm}[1]{\mymatrix{#1}}
\newcommand{\Rp}{\matrixcol{R}{p}}
\newcommand{\Rl}{\matrixrow{R}{\ell}}

\newcommand{\reals}{\ensuremath{\mathbb{R}}}
\newcommand{\nonnegativereals}{\ensuremath{\mathbb{R}_{\geq 0}}}

\newcommand{\Event}{\mathcal{E}}
\newcommand{\Ind}[1]{\mathbb{I}_{\left\{#1\right\}}}

\def\vector#1{\vec{#1}}
\def\hvec#1{\vec{\hat{#1}}}
\def\tvec#1{\vec{\tilde{#1}}}
\def\La{\Lambda}
\def\la{\lambda}

\newcommand{\REMOVED}[1]{}

\title{Robust Line Planning in case of Multiple Pools and Disruptions}

\author{%
Apostolos Bessas \inst{1,3}
\and
Spyros Kontogiannis \inst{1,2}
\and Christos Zaroliagis \inst{1,3}
}

\institute{R.A.~Computer Technology Institute, N.~Kazantzaki Str.,
Patras University Campus, 26500 Patras, Greece
\\
\and
Computer Science Department, University of Ioannina,
45110 Ioannina, Greece
\\
\and
Department of Computer Engineering and Informatics, University of
Patras, \\ 26500 Patras, Greece
\\
Email: {\tt mpessas@ceid.upatras.gr}, {\tt kontog@cs.uoi.gr}, {\tt zaro@ceid.upatras.gr}
}

\maketitle
\vspace*{-0.5cm}
\begin{center}
\today
\end{center}
\vspace*{-0.5cm}
\begin{abstract}
We consider the line planning problem in public 
transportation, under a robustness perspective. 
We present a mechanism for robust line planning in the case of multiple
line pools, when the line operators have a different utility function per
pool. We conduct an experimental study of our mechanism on both synthetic
and real-world data that shows fast convergence to the optimum. We also
explore a wide range of scenarios, varying from an arbitrary initial state
(to be solved) to small disruptions in a previously optimal solution (to
be recovered). Our experiments with the latter scenario show that our
mechanism can be used as an online recovery scheme causing the system to
re-converge to its optimum extremely fast.
\end{abstract}

\section{Introduction}
\label{sec:intro}

Line planning is an important phase in the hierarchical planning process
of every railway (or public transportation) network\footnote{For the sake of convenience, we
concentrate in this work on railway networks, but the methods and
ideas developed can be applied to any public transportation network.}.
The goal is to determine the routes (or lines) of trains that will serve the
customers along with the frequency each train will serve a particular
route. Typically, the final set of lines is chosen by a (predefined) set of candidate
lines, called the \emph{line pool}. In certain cases, there may be
\emph{multiple line pools} representing the availability of
the network infrastructure at different time slots or zones. This is due to
variations in customer traffic (e.g., rush-hour pool, late evening pool,
night pool), maintenance (some part of the network at a specific time zone
may be unavailable), dependencies between lines (e.g., the choice of a high-speed
line may affect the choice of lines for other trains), etc.

The line planning problem has been extensively studied under cost-oriented
or customer-oriented approaches (see e.g., \cite{D78,GHK2004,SS2005,S2005}).
Recently, robustness issues have been started to be investigated. In
the \emph{robust line planning} problem, the task is to provide a set of lines
along with their frequencies, which are robust to fluctuations of input
parameters; typical fluctuations include, for instance, disruptions to daily
operations (e.g., delays), or varying customer demands. In \cite{SS2006},
a game-theoretic approach to robust line planning was presented that
delivers lines and frequencies that are robust to delays.

A different perspective of robust line planning was investigated in \cite{BKZ2009}.
This perspective stems form recent regulations in the European Union
that introduce competition and free railway markets for its members.
Under these rules the following scenario emerges: there is a (usually state)
authority that manages the railway network infrastructure, referred to as the
\emph{Network Operator} (NOP), and a (potentially) large number of \emph{Line Operators}
(LOPs) operating as commercial organizations which want to offer services to their
customers using the given railway network. These LOPs act as competing
agents for the exploitation of the shared infrastructure and are unwilling to
disclose their utility functions that demonstrate their true incentives.
The network operator wishes to set up a fair cost sharing scheme for the
usage of the shared resources and to ensure the maximum possible level
of satisfaction of the competing agents (by maximizing their aggregate utility
functions). The former implies a resource pricing scheme that is robust against
changes in the demands of the LOPs, while the latter establishes a notion of a
socially optimal solution, which could also be considered as a fair solution,
in the sense that the average level of satisfaction is maximized. In other
words, the NOP wishes to establish an incentive-compatible mechanism that
provides \emph{robustness} to the
system in the sense that it tolerates the agents' unknown incentives and elasticity
of demand requests and it eventually stabilizes the system at an equilibrium point
that is as close as possible to the social optimum.

The first such mechanism, for robust line planning in the aforementioned scenario,
was presented in \cite{BKZ2009}. In that paper, the following mechanism was investigated
(motivated by the pioneering work of Kelly et al.~\cite{K97,KMT98} in communication networks):
the LOPs offer bids, which they (dynamically) update for buying frequencies. The NOP
announces an (anonymous) resource pricing scheme, which indirectly implies an allocation
of frequencies to the LOPs, given their own bids.
For the case of a single pool of lines, a distributed, dynamic, LOP bidding and (resource)
price updating scheme was presented, whose equilibrium point is the unknown social optimum --
assuming strict concavity and monotonicity of the private (unknown) utility functions.
This development was complemented by an experimental study on a discrete variant of the distributed,
dynamic scheme on both synthetic and real-world data showing that the mechanism converges really
fast to the social optimum. The approach to the single pool was extended to derive an analogous
mechanism for the case of multiple line pools, where it was assumed that (i) the NOP
can periodically exploit a whole set of (disjointly operating) line pools and
he decides on how to divide the whole infrastructure among the different pools so that the
resource capacity constraints are preserved; (ii) each LOP may be interested in different lines
from different pools; and (iii) each LOP has a single utility function which depends on the aggregate
frequency that she gets from all the pools in which she is involved. 
%

The aforementioned theoretical framework demonstrated the potential of converging to the social
optimum via a mechanism (dynamic system) that exploits the selfishness of LOPs. A significant
issue is the speed or rate of convergence of this mechanism. Since there was no theoretical treatment
of this issue, its lack was covered in \cite{BKZ2009} for the single pool case via a complementary experimental 
study. Despite, however, the significance of the convergence rate issue, the mechanism for the multiple pool case 
was \emph{not} experimentally evaluated in \cite{BKZ2009}.

For the case of multiple line pools, it is often more realistic to assume that each LOP has a different
utility function per pool, since different pools are expected to provide different profits (e.g., 
intercity versus regional lines, or rush-hour versus late-evening lines). Moreover, it seems
more natural to assume that each LOP has a different utility function per pool that depends on
the frequency she gets for that pool, rather than a single utility function that depends on the
total frequency she gets across all pools.

In this work, we continue this line of research by further investigating the multiple pool
case. In particular, we make the following contributions:
(1) Contrary to the approach in \cite{BKZ2009}, we consider the case where
each LOP has a different utility function for each line pool she is interested in, and show
how the approach in \cite{BKZ2009} can be extended in order to provide a mechanism 
for this case, too.
(2) We conduct an experimental study on a discrete variant of the new mechanism on both synthetic
and real-world data demonstrating its fast convergence to the social optimum.
(3) We conduct an additional experimental study, on both synthetic
and real-world data, to investigate the robustness of the system in
the case of disruptions that affect the available capacity, which may be
reduced (due to temporary unavailability of part of the network), or
increased (by allowing usage of additional infrastructure during certain busy
periods). In this case, we show that the NOP can re-converge (recover) the system to
the social optimum pretty fast, starting from a previous optimal solution.

The rest of the paper is organized as follows. In Section~\ref{sec:mlp},
we present the mechanism for the case of multiple line pools, where each
LOP has a different utility function per pool. In Sections~\ref{sec:exp-mlp}
and \ref{sec:disruptions}, we present the main contributions of this work
for the multiple line pool case; namely, the experimental study showing the
fast convergence of the new mechanism, and the robustness
of the system in the case of disruptions, respectively.
We conclude in Section~\ref{sec:conclusions}.

\section{Multiple Line Pools: Different Utility Functions per Pool}
\label{sec:mlp}


The exposition in this section follows that in \cite{BKZ2009}.
Although the analysis of our approach is based on the same
methodology as that in \cite{BKZ2009}, 
we have chosen to present all the details for making the paper
self-contained and for the better understanding of the experimental studies.

In the \emph{line planning} problem, the NOP provides the public transportation
network in the form of a directed graph $G = (V, L)$, where $V$ is the
node set representing train stations and important railway junctions,
and $L$ is the edge set representing direct connections or links (of
railway tracks) between elements of $V$. Each edge $\ell \in L$ is
associated with a capacity $c_{\ell} > 0$, which limits the
number of trains that can use this edge in the period examined.
A line $p$ is a path in $G$.

For our purposes, we assume that there is set $K$ of line pools, where each pool represents a different set of possible routes (each corresponding to a different period of the day). We envision the line pools to be implemented in disjoint time intervals (e.g., via some sort of time division multiplexing), and also to concern different characteristics of the involved lines (e.g., high-speed pool, regular-speed pool, local-trains pool, rush-hour pool, night-shift pool, etc.). The capacity of each resource (edge) refers to its usage (number of trains) over the whole time period we consider
(e.g., a day), and if a particular pool consumes (say) 50\% of the whole
infrastructure, then this implies that for all the lines in this pool,
each resource may exploit at most half of its capacity. It is up to the NOP to determine how to split a whole operational period of the railway infrastructure among the different pools, so that (for the whole period) the resource capacity constraints are not violated.

There is also a set $P$ of LOPs, who choose their lines from $K$. We assume that
each LOP $p \in P$ is interested only in one line in each pool
(we can always enforce this assumption by considering a LOP interested
in more than one routes as different LOPs distinguished by the specific route).
Each line pool and the preference of LOPs to lines in it are represented by a
\emph{routing matrix} $\mm{R}(k) \in \{0, 1\}^{|L| \times |P|}, k \in K$. Each row $\matrixrow{R}{\ell}(k)$ corresponds to a different edge $\ell \in L$,
and each column $\matrixcol{R}{p}(k)$ corresponds to a different
LOP $p \in P$, showing which edges comprise her line in pool $k$.

Each LOP $p \in P$ acquires a \emph{frequency} of trains that she wishes to route over her paths in $\matrixcol{R}{p}(k), k \in K$, such that no edge capacity constraint is violated by the aggregate frequency running through it by all LOPs and pools. A utility function $U_{p,k} \colon \nonnegativereals \mapsto \nonnegativereals$ determines the \emph{level of satisfaction} of LOP $p \in P$ in pool $k \in K$ for being given an end-to-end frequency $x_{p,k} > 0$.
Having different utility functions per pool, instead of a single utility function across all pools, is more generic and hence more realistic, since a LOP $p$ can indeed have different valuations for different periods of a day (rush-hour pool vs night-shift pool) and/or different types of trains (high-speed pool vs local-trains pool). These utility functions are assumed to be strictly increasing, strictly
concave, non-negative real functions of the end-to-end frequency
$x_{p,k}$ allocated to LOP $p \in P$ in pool $k \in K$.

The aggregate satisfaction level $U_p$ of LOP $p \in P$ across all
pools is given by the sum of the individual gains she has in
each pool, that is
\[
U_p (\vector{x}_p) = U_p (x_{p,1},\ldots,x_{p,k}) = \sum_{k \in K} U_{p,k} (x_{p,k}), \quad \forall p \in P
\]
where $\vector{x}_p =  (x_{p,k})_{k \in K}$ is the vector of frequencies that $p$ gets for all the pools. The utility functions are \emph{private} to the LOP; she is not willing to share them for competitiveness reasons, not even with the NOP. This has a few implications on the necessary approach to handle the problem.

The NOP, on the other hand, wishes to allocate to each LOP a frequency
vector $\hvec{x}_p = \sum_{k \in K} \hat{x}_{p,k}$ such that the cumulative satisfaction of 
all the LOPs is maximized, while respecting all the edge capacity constraint. To achieve this, 
the NOP divides the whole railway infrastructure to the pools, using variables $f_k, k \in K$ that determine the
proportion of the total capacity of the edges that is assigned to pool
$k$. Hence, the NOP wishes to solve the following strictly convex optimization
problem:
 \begin{equation}
  \label{eq:MSC-diffU}
  \tag{MSC-II}
  \begin{aligned}
    \max \quad & \sum_{p \in P}U_p(\vector{x}_p) = \sum_{p \in P} \sum_{k \in
      K}U_{p,k} \left(x_{p,k} \right)\\
    s.t. \quad & \sum_{p \in P}R_{\ell, p}(k) \cdot x_{p,k} \le c_{\ell}
    \cdot f_k, \, \forall (\ell, k) \in L \times K\\
    & \sum_{k \in K} f_k \le 1\\
    & \mv{x}, \mv{f} \ge 0
  \end{aligned}
\end{equation}
\REMOVED{
The problem is convex, since it is a maximization problem and the
objective function is the sum of concave functions as well as because
the feasible space is linear (we make the assumption that a convex
problem is a minimization problem with a convex objective function,
which is equivalent with a maximization problem with a concave
objective function).
}
Clearly, the NOP cannot solve this problem directly for (at least) two reasons: (i) the utility functions are unknown to him; (ii) the scale of the problem can be too large (as it is typical with railway networks) so that it can be solved efficiently via a centralized computation. The latter is particularly important when the whole system is already at some equilibrium state and then suddenly a (small, relative to the size of the whole problem) perturbation in the problem parameters occurs. Then, rather than having a whole new re-computation of the new optimal solution from scratch, it is particularly desirable that a dynamical scheme allows convergence to the new optimal solution, starting from this worm start (of the previously optimal solution). All the above reasons dictate searching for a different solution approach, that has to be as decentralized as possible.


\paragraph{\textbf{Optimal Solution.}}
We adopt the approach in \cite{BKZ2009} to design a mechanism that will be
run by the NOP in order to solve the above problem. In particular, rather than having the NOP directly deciding for the frequencies of all the LOPs in each pool, we first let each LOP make her own bid for frequency in each pool. Then, the NOP consider the solution of a convex program which is similar, but not identical to (\ref{eq:MSC-diffU}), using a set of (strictly increasing, strictly concave) pseudo-utilities. Our goal is to exploit the rational (competitive) behavior of the LOPs, in order to assure that eventually the optimal solution reached for this new program is identical to that of (\ref{eq:MSC-diffU}), as required.

We start by describing the necessary and sufficient optimality (KKT) conditions for (\ref{eq:MSC-diffU}). We study the Lagrangian function and exploit the economic interpretation of the Lagrangian multipliers. Let $\vector{\La}$
be the vector of Lagrangian multipliers for the resource capacity constraints,
and $\zeta$ the Lagrangian multiplier concerning the constraint for the
capacity proportions per pool. The Lagrangian function of (\ref{eq:MSC-diffU})
is as follows:
\begin{equation*}
  \label{eq:lagrange-MNET-II}
  \begin{aligned}
    & L(\mv{x}, \mv{f}, \mv{\Lambda}, \zeta) = \\
    & \sum_{p\in P}\sum_{k \in K} U_{p,k}(x_{p,k}) - \sum_{\ell \in
      L}\sum_{k\in K} \Lambda_{\ell,k}\cdot \left[ \sum_{p\in
        P}R_{\ell,p}(k) \cdot x_{p,k} - c_{\ell}\cdot f_k \right] -
    \zeta\left[\sum_{k\in K}
      f_k - 1\right] =\\
    & \sum_{p\in P}\sum_{k \in K}\left[ U_{p,k}(x_{p,k}) - x_{p,k}
      \left( \sum_{\ell\in L} \Lambda_{\ell,k}\cdot R_{\ell,p}(k)
      \right) \right] + \sum_{k\in K}
    f_k\cdot\left[\mv{c}^T\matrixcol{\Lambda}{k} - \zeta \right] +
    \zeta =\\
    & \sum_{p\in P}\sum_{k \in K} \left[ U_{p,k}(x_{p,k}) - x_{p,k}
      \cdot \mu_{p,k}(\mv{\Lambda}) \right] + \sum_{k\in K} f_k \cdot
    \left[\mv{c}^T\matrixcol{\Lambda}{k} - \zeta \right] + \zeta,
  \end{aligned}
\end{equation*}
where $\mu_{p,k}(\mv{\Lambda}) \equiv \sum_{\ell \in L}\Lambda_{\ell,k} \cdot R_{\ell,p}(k)$. We interpret $\Lambda_{\ell,k}$ to be the per-unit-of-frequency price of edge $\ell \in L$ in pool $k \in K$. Thus, $\mu_{p,k}(\mv{\Lambda})$ is the end-to-end per-unit price that LOP $p$ has to pay for pool $k$.

The optimality (KKT) conditions for (\ref{eq:MSC-diffU}) are the following:
\begin{align}
  \forall (p,k)\in P\times K, \;
  U'_{p,k}(\hat{x}_{p,k}) & = \hat{\mu}_{p,k} \equiv
  \mu_{p,k}(\hvec{\Lambda})
  \tag{KKT-MSC-II.a}\label{eq:selfish-KKT-MSC-2} \\
  \mv{c}^T\cdot \matrixcol{\hat{\Lambda}}{k} \equiv \sum_{\ell\in L}
  \hat{\Lambda}_{\ell,k}\cdot c_{\ell} & = \hat{\zeta}, \; \forall
  k\in K \tag{KKT-MSC-II.b}\label{eq:zeta-KKT-MSC-2}\\
  \hat{\Lambda}_{\ell,k} \left[\sum_{p\in P}R_{\ell,p}(k)
    \hat{x}_{p,k} - c_\ell \hat{f}_k \right] & = 0, \; \forall
  (\ell,k)\in L \times K
  \tag{KKT-MSC-II.c}\label{eq:complementarity1-KKT-MSC-2} \\
  \hat{\zeta}\cdot\left(\sum_{k\in K} \hat{f}_k - 1\right) & = 0
  \tag{KKT-MSC-II.d}\label{eq:complementarity2-KKT-MSC-2}\\
  \sum_{p\in P} R_{\ell,p}(k)\cdot \hat{x}_{p,k} & \leq
  c_\ell\cdot\hat{f}_k, \; \forall (\ell,k)\in L\times K
  \tag{KKT-MSC-II.e} \\
  \sum_{k\in K} \hat{f}_k & \leq 1\tag{KKT-MSC-II.f}\\
  \hvec{x}, \hvec{f}, \hvec{\Lambda}, \hat{\zeta} & \geq \mv{0} \nonumber
\end{align}
These conditions characterize the optimal solution of \eqref{eq:MSC-diffU}. From these equations, one can observe for the optimal solution $(\hvec{x},\hvec{f})$ and its accompanying vector of Lagrangian multipliers $(\hvec{\La},\hat{\zeta})$:
\begin{enumerate}

\item By equation~(\ref{eq:zeta-KKT-MSC-2}) all pools (at optimality) have the same (weighted with the capacities of the edges) aggregate per-unit-of-frequency cost $\hat{\zeta}$.

\item By equation~(\ref{eq:complementarity2-KKT-MSC-2}) either
  $\hat{\zeta} = 0$ or $\sum_{k \in K}\hat{f}_k = 1$. The former case
  (all prices are set to zero)
  cannot actually happen, since otherwise all the pools would have zero aggregate cost (i.e., every edge is underutilized in every pool). Then we could slightly increase the frequencies of all the LOPs in all the pools by a positive constant and keep exactly the same vector of proportions $\hvec{f}$, which would lead to a strictly better solution (the utility functions are strictly increasing), contradicting the optimality of $\hvec{x}$. Hence, we consider only the case, where at least some resources have positive costs. In this case, the condition $\sum_{k \in K}\hat{f}_k = 1$ must hold. Any vector \mv{f} for which this condition holds is said to \emph{completely divide} the infrastructure among the pools.

\end{enumerate}
%

\paragraph{\textbf{Limited Information and Decentralized Mechanism.}}
As already mentioned, the NOP cannot know the utility functions. Therefore, a slightly different problem is considered: Each LOP $p\in P$ announces (non-negative) bids $w_{p,k}\geq 0$ (one per pool), which she is committed to spend for acquiring frequencies in the pools. Consequently, the NOP replaces
the unknown utility functions with the pseudo-utilities $w_{p,k}\log(x_{p,k})$ in order to determine a frequency vector that maximizes the aggregate level of pseudo-satisfaction. Observe that these used pseudo-utilities are also strictly increasing, strictly concave functions of the LOPs' frequencies. That means, that NOP wishes to solve the following (strictly convex) optimization problem that is completely known to him:
\begin{equation}
  \label{eq:MNET-II}
  \tag{MNET-II}
  \begin{aligned}
    \max \quad & \sum_{p\in P}\sum_{k \in K} w_{p,k}\log(x_{p,k})\\
    s.t. \quad & \sum_{p\in P}R_{\ell,p}(k) \cdot x_{p,k} \leq
    c_{\ell}\cdot f_k, \; \forall (\ell,k)\in L\times K\\
    & \sum_{k\in K} f_k \leq 1 \\
    & \mv{x}, \mv{f} \ge \mv{0}.
  \end{aligned}
\end{equation}
The optimal solution vector $(\bvec{x},\bvec{f})$, along with the corresponding vector of Lagrange multipliers $(\bvec{\La},\bar{\zeta})$, must satisfy the following necessary and sufficient optimality (KKT) conditions for \eqref{eq:MNET-II}:
\begin{align*}
  \forall (p,k)\in P\times K, \; \frac{w_{p,k}}{\bar{x}_{p,k}} & =
  \bar{\mu}_{p,k} \equiv \mu_{p,k}(\bvec{\Lambda})
  \tag{KKT-MNET-II.a}\label{eq:selfish-KKT-MNET-2} \\
  \mv{c}^T\cdot \matrixcol{\bar{\Lambda}}{k} \equiv \sum_{\ell\in L}
  \bar{\Lambda}_{\ell,k}\cdot c_{\ell} & = \bar{\zeta}, \; \forall
  k\in K \tag{KKT-MNET-II.b}\label{eq:zeta-KKT-MNET-2}\\
  \bar{\Lambda}_{\ell,k} \left[\sum_{p\in P}R_{\ell,p}(k)
    \bar{x}_{p,k} - c_\ell \bar{f}_k \right] & = 0, \; \forall
  (\ell,k)\in L \times K
  \tag{KKT-MNET-II.c}\label{eq:complementarity1-KKT-MNET-2} \\
  \bar{\zeta}\cdot\left(\sum_{k\in K} \bar{f}_k - 1\right) & = 0
  \tag{KKT-MNET-II.d} \label{eq:complementarity2-KKT-MNET-2}\\
  \sum_{p\in P} R_{\ell,p}(k)\cdot \bar{x}_{p,k} & \leq
  c_\ell\cdot\bar{f}_k, \; \forall (\ell,k)\in L\times K
  \tag{KKT-MNET-II.e} \\
  \sum_{k\in K} \bar{f}_k & \leq 1  \tag{KKT-MNET-II.f}\\
  \bvec{x}, \bvec{f}, \bvec{\Lambda}, \bar{\zeta} & \geq \mv{0}
  \nonumber\\
\end{align*}

This problem can of course be solved in polynomial time, given the bid of the
LOPs $\vector{w} = (w_{p,k})_{(p,k)\in P\times K}$. Equation~\eqref{eq:selfish-KKT-MNET-2} means that the NOP (correctly) assigns
frequency $\bar{x}_{p,k} = \frac{w_{p,k}}{\bar{\mu}_{p,k}}$, which is affordable at the (path) per-unit price $\bar{\mu}_{p,k}$ for LOP $p$ who is committed to afford an amount of $w_{p,k}$ for acquiring frequency in pool $k$. As already said, $(\bvec{x},\bvec{f})$ is the optimal solution for any bid vector declared by the LOPs, and in particular it also holds for the true bid vector that the LOPs would really wish to afford.

Observe, also, that problems (KKT-MSC-II) and (KKT-MNET-II) differ
only in equations (\ref{eq:selfish-KKT-MSC-2}) and (\ref{eq:selfish-KKT-MNET-2}). Our next step is to somehow assure that these two conditions coincide. To this direction, we exploit the rational behavior of the LOPs: Each LOP wishes to maximize her own aggregate level of satisfaction, therefore, she would declare a bid vector that would actually achieve this. In what follows, we assume that the LOPs are \emph{price takers} meaning that each of them considers the prices announced by the NOP as constants, with no hope of affecting them by their own bid vector. This property is important in the following analysis, and is realistic when there exist many LOPs, each controlling only negligible fractions of the total flow (or bidding process) in the system. We now state and prove an important theorem for the existence of a
mechanism for this problem.

\begin{theorem}
  Given a transportation network $G = (V, L)$, a set of line pools $K$
  and a set $P$ of selfish, price-taking LOPs, each having a private
  utility function for each pool with parameter the frequency that is
  allocated to her in the particular pool, there is a mechanism (a pair
  of a frequency allocation mechanism and a resource pricing scheme) that computes in polynomial time the optimal solution of the sum of the utility functions of the players, while respecting the capacities of the edges.
\end{theorem}
\begin{proof}
  We create the following mechanism. The NOP completely divides the
  infrastructure and, depending on the bids of the LOPs, announces the
  optimal frequencies $\bar{x}_{p,k}$ for each LOP $p \in P$ and the current
  resource (edge) prices $\bar{\Lambda}_{\ell, k}, \forall \ell \in L, k \in
  K$. We have that, $\bar{x}_{p,k} = \frac{w_{p,k}}{\bar{\mu}_{p,k}}$. Each (selfish) LOP is interested in solving the following strictly convex optimization problem:
  \begin{equation}
    \label{eq:muser-II}
    \tag{MUSER-II}
  \begin{aligned}
    \max \quad & \sum_{k \in K}\left(U_{p,k}(\bar{x}_{p,k}) - w_{p,k}\right)
    = \sum_{k \in
      K}\left(U_{p,k}\left(\frac{w_{p,k}}{\bar{\mu}_{p,k}}\right) -
      w_{p,k}\right)\\
    s.t. \quad & \hspace{8em} w_{p,k} \ge 0, \; \forall k \in K
  \end{aligned}
  \end{equation}
  That is, she tries to maximize the utility she gets over the frequency
  she will be granted by the NOP, after subtracting the bid she has to pay for this frequency.

  Due to the price taking property of LOP $p$, we can obtain the solution of the \eqref{eq:muser-II} by differentiating the objective function and setting every partial derivative equal to zero:
  \begin{equation*}
    \frac{\partial}{\partial w_{p,k}} \sum_{k \in K} \left(
      U_{p,k} \left( \frac{w_{p,k}}{\bar{\mu}_{p,k}}\right) - w_{p,k} \right)
    = 0, \; \forall k \in K
  \end{equation*}
  we get
  \begin{gather*}
    U'_{p,k}\left(\frac{w_{p,k}}{\bar{\mu}_{p,k}}\right) \cdot
    \frac{1}{\bar{\mu}_{p,k}} - 1 = 0 \Rightarrow
    U'_{p,k}\left(\bar{x}_{p,k}\right) = \bar{\mu}_{p,k}
  \end{gather*}

  However, this is the set of equations \eqref{eq:selfish-KKT-MSC-2}
  of the KKT conditions for the (MSC-II) problem. That is, every LOP
  chooses the bids in a way that matches the equation
  \eqref{eq:selfish-KKT-MSC-2}. But this will happen in the optimal
  solution of (MNET-II) as well. Therefore, the selfish choice of bids by the LOPs enforces the identification of the optimal solution of MSC-II and MNET-II. I.e., we have managed to delegate the problem of choosing the optimal bids in each iteration to the LOPs and, in doing so, we have matched
  the first equation of the optimality conditions of MSC-II.
\qed
\end{proof}

This polynomially tractable mechanism (based on the solvability of MNET-II) is totally centralized and rather inconvenient for a dynamically changing (over time), large-scale railway system. The following lemma (whose proof can be found in the Appendix) is crucial in deriving a dynamical system for solving \eqref{eq:MSC-diffU}.

\begin{lemma}
\label{lm:main}
  For any (fixed) vector \mv{f} of capacity proportions that completely divides the railway infrastructure among the pools, the optimal value of \eqref{eq:MSC-diffU} exclusively depends on the optimal vector \bvec{\Lambda} of the per-unit-of-frequency prices of the resources.
\end{lemma}

\REMOVED{
\begin{proof}
  The Lagrangian function of the problem is:
  \begin{equation*}
    L(\mv{x}, \mv{f}, \mv{\Lambda}, \zeta) = \\
    \sum_{p\in P}\sum_{k \in K} \left[ U_{p,k}(x_{p,k}) - x_{p,k}
      \cdot \mu_{p,k}(\mv{\Lambda}) \right] + \sum_{k\in K} f_k \cdot
    \left[\mv{c}^T\matrixcol{\Lambda}{k} - \zeta \right] + \zeta.
  \end{equation*}
  The dual problem of \eqref{eq:MSC-diffU} is the following:
  \begin{equation*}
    \begin{aligned}
      \max \quad & D(\mv{\Lambda}, \zeta)\\
      s.t. \quad & \Lambda_{\ell,k} \ge 0 \quad \forall \ell \in L, \,
      \forall k \in K\\
      & \zeta \ge 0,
    \end{aligned}
  \end{equation*}
  where
  \begin{align*}
    & D(\mv{\Lambda}, \zeta)  =\\
    & = \max\left\{L(\mv{x}, \mv{f}, \mv{\Lambda}, \zeta) \colon
      \mv{x}, \mv{f} \ge \mv{0} \right\} \\
    & = \max_{\mv{x}, \mv{f} \ge \mv{0}} \left\{ \sum_{p \in P}\sum_{k
        \in K}\left[ U_{p,k}(x_{p,k}) - x_{p,k}\sum_{l \in
          L}\Lambda_{\ell, k}R_{\ell, p}(k) \right] + \sum_{k \in
        K}f_k \left[ \sum_{\ell \in L}\Lambda_{\ell,k}c_{\ell} - \zeta
      \right] +
      \zeta \right\}\\
    & = \max_{\mv{x} \ge \mv{0}} \left\{ \sum_{p \in P}\sum_{k \in
        K}\left[ U_p(x_p) - x_{p,k}\mu_{p,k}(\mv{\Lambda}) \right]
    \right\} + \max_{\mv{f} \ge \mv{0}} \left\{ \sum_{k \in K}f_k
      \left[ \sum_{\ell \in L}\Lambda_{\ell, k}c_{\ell} - \zeta
      \right] \right\} + \zeta.
  \end{align*}

  For the part
  \begin{equation*}
    P_1(\mv{\Lambda}) = \max_{\mv{x} \ge \mv{0}} \left\{ \sum_{p \in
        P} \left[ U_p(x_p) - \sum_{k \in K}x_{p,k} \cdot
        \mu_{p,k}(\mv{\Lambda}) \right] \right\}
  \end{equation*}
  we have, that it is a simple maximization problem the solution of
  which depends exclusively on the optimal vector of resource prices
  \bvec{\Lambda}.

  For the part
  \begin{align*}
    P_2(\mv{\Lambda}, \zeta) & = \max_{\mv{f} \ge \mv{0}} \left\{
      \sum_{k \in K} f_k \cdot \left[ \sum_{\ell \in
          L}\Lambda_{\ell,k} \cdot c_{\ell} - \zeta \right] \right\}
    + \zeta\\
    & = \max_{\mv{f} \ge \mv{0}} \left\{ \sum_{k \in K} \cdot \left(
        \mv{c}^T \cdot \matrixcol{\Lambda}{k} \right) + \zeta \cdot
      \left( 1 - \sum_{k \in K}f_k \right) \right\}
  \end{align*}
  we have, that at optimality $\zeta \cdot \left( 1 - \sum_{k \in
      K}f_k \right)$ equals to zero. So, we get:
  \begin{align*}
    P_2(\mv{\Lambda}, \zeta) & = \max_{\mv{1}^T\mv{f} = 1 \colon
      \mv{f} \ge \mv{0}} \left\{ \sum_{k \in K}f_k \cdot \left(
        \mv{c}^T \cdot \matrixcol{\Lambda}{k} \right) \right\}\\
    & = \max_{k \in K} \left\{ \mv{c}^T \cdot \matrixcol{\Lambda}{k}
    \right\} \\
    & = \min \left\{ \mv{z} \colon \mv{z} \cdot \mv{1}^T \ge \mv{c}^T \cdot
      \mv{\Lambda} \right\}.
  \end{align*}
  That is, $P_2(\mv{\Lambda}, \zeta)$ calculates the maximum cost
  among the pools. This, however, depends only on the optimal vector
  of resource prices \bvec{\Lambda}.
\end{proof}

So, to prove the convergence of the above system , we can use the
Lyapunov function
\begin{equation*}
  V(\mm{\Lambda}(t)) = \frac{1}{2} \cdot \left( \mm{\Lambda}(t) -
    \bvec{\Lambda} \right)^T \cdot \left( \mm{\Lambda}(t) -
    \bvec{\Lambda} \right)
\end{equation*}
and show that $\frac{dV(\mv{\Lambda}(t))}{dt} \le 0$.

Indeed,
\begin{align*}
  \frac{dV(\mm{\Lambda}(t))}{dt} & =\\
  & = \sum_{\ell\in L}\sum_{k\in K} \left(
    \Lambda_{\ell,k}(t)-\bar{\Lambda}_{\ell,k} \right) \cdot
  \dot{\Lambda}(t) \\
  & = \sum_{\ell\in L}\sum_{k\in K} \left(
    \Lambda_{\ell,k}(t)-\bar{\Lambda}_{\ell,k} \right) \cdot \left[
    \max\left\{0, y_{\ell,k}(t) - c_{\ell}\bar{f}_k \right\} \cdot
    \Ind{\Lambda_{\ell,k}(t)=0} \right. \\
  & \left.\hspace*{5cm} + \left(y_{\ell,k}(t) -
      c_{\ell}\bar{f}_k\right) \cdot \Ind{\Lambda_{\ell,k}(t)>0}
  \right]\\
  & \leq \sum_{\ell\in L}\sum_{k\in K} \left( \Lambda_{\ell,k}(t) -
    \bar{\Lambda}_{\ell,k} \right) \cdot \left[ y_{\ell,k}(t) -
    c_{\ell} \bar{f}_k \right] \quad\text{(from
    \ref{inequality4})} \\
  & = \sum_{\ell\in L}\sum_{k\in K} \left( \Lambda_{\ell,k}(t) -
    \bar{\Lambda}_{\ell,k} \right) \cdot \left[ y_{\ell,k}(t) -
    \bar{y}_{\ell,k} + \bar{y}_{\ell,k} - c_{\ell} \bar{f}_k
  \right]\\
  & \leq \sum_{\ell\in L}\sum_{k\in K} \left( \Lambda_{\ell,k}(t) -
    \bar{\Lambda}_{\ell,k} \right) \cdot \left[ y_{\ell,k}(t) -
    \bar{y}_{\ell,k} \right] \quad\text{(from
    \ref{inequality5})}\\
  & = \sum_{\ell\in L}\sum_{k\in K} \left( \Lambda_{\ell,k}(t) -
    \bar{\Lambda}_{\ell,k} \right) \cdot \sum_{p\in P} R_{\ell,p}(k)
  \cdot \left[ x_{p,k}(t) - \bar{x}_{p,k} \right]\\
  & = \sum_{p\in P}\sum_{k\in K} \left[ x_{p,k}(t) - \bar{x}_{p,k}
  \right] \cdot \left( \mu_{p,k}(t) - \bar{\mu}_{p,k} \right)\\
  & = \sum_{p\in P}\sum_{k\in K} \left[ x_{p,k}(t) - \bar{x}_{p,k}
  \right] \cdot \left( \frac{\bar{w}_{p,k}}{x_{p,k}(t)} -
    \frac{\bar{w}_{p,k}}{\bar{x}_{p,k}} \right)\\
  & = \sum_{p\in P}\sum_{k\in K}\bar{w}_{p,k} \cdot \left[2 -
    \frac{x_{p,k}(t)}{\bar{x}_{p,k}} -
    \frac{\bar{x}_{p,k}}{x_{p,k}(t)} \right]\\
  & \leq 0 \quad\text{(from \ref{inequality6})}
\end{align*}

In the above, we have used the following inequalities:
\begin{itemize}
\item \label{inequality4} The following inequality holds:
  \begin{align*}
    & \sum_{\ell\in L}\sum_{k\in K} \left(
      \Lambda_{\ell,k}(t)-\bar{\Lambda}_{\ell,k} \right) \cdot \left[
      \max\left\{0, y_{\ell,k}(t) - c_{\ell}\bar{f}_k \right\} \cdot
      \Ind{\Lambda_{\ell,k}(t)=0} \right. \\
    & \left.\hspace*{5cm} + \left(y_{\ell,k}(t) -
        c_{\ell}\bar{f}_k\right) \cdot \Ind{\Lambda_{\ell,k}(t)>0}
    \right]\\
    & \leq \sum_{\ell\in L}\sum_{k\in K} \left( \Lambda_{\ell,k}(t) -
      \bar{\Lambda}_{\ell,k} \right) \cdot \left[ y_{\ell,k}(t) -
      c_{\ell}
      \bar{f}_k \right]\\
  \end{align*}
  This is clearly true $ \forall (\ell, k) \colon \Lambda_{\ell,k}(t)
  > 0$, αλλά και $ \forall (\ell, k) \colon \Lambda_{\ell,k}(t) =
  0$. The last part holds, because it will either be $y_{\ell,k}(t) -
  c_{\ell}\bar{f}_{k} \ge 0$ (then it is obvious) or
  $y_{\ell,k}(t) - c_{\ell}\bar{f}_k < 0$. In this case:
  \begin{equation*}
    \left( \Lambda_{\ell,k}(t) - \bar{\Lambda}_{\ell,k} \right) \cdot
    \max\left\{0, y_{\ell,k}(t) - c_{\ell}\bar{f}_k \right\} = -
    \bar{\Lambda}_{\ell,k} \cdot 0 < -\bar{\Lambda}_{\ell,k} \cdot
    \left[ y_{\ell,k}(t) - c_{\ell}\bar{f}_{k} \right].
  \end{equation*}
\item \label{inequality5} The following inequality holds:
  \begin{align*}
    & \sum_{\ell\in L}\sum_{k\in K} \left( \Lambda_{\ell,k}(t) -
      \bar{\Lambda}_{\ell,k} \right) \cdot \left[ y_{\ell,k}(t) -
      \bar{y}_{\ell,k} + \bar{y}_{\ell,k} - c_{\ell} \bar{f}_k
    \right]\\
    & \leq \sum_{\ell\in L}\sum_{k\in K} \left( \Lambda_{\ell,k}(t) -
      \bar{\Lambda}_{\ell,k} \right) \cdot \left[ y_{\ell,k}(t) -
      \bar{y}_{\ell,k} \right]\\
  \end{align*}
  The second inequality holds, because for the optimal matrix of
  resource prices \bvec{\Lambda} (for given \bvec{w} and \bvec{f}) it
  holds, that
  \begin{equation*}
    \sum_{\ell \in L}\sum_{k \in K} \bar{\Lambda}_{\ell,k}
    (\bar{y}_{\ell,k} - c_{\ell}\bar{f}_k) = 0
  \end{equation*}
  (due to equation~\eqref{eq:complementarity1-KKT-MSC-2}).
\item \label{inequality6} The following inequality holds:
  \begin{equation*}
    \sum_{p\in P}\sum_{k\in K}\bar{f}_k\bar{w}_p \cdot \left[2 -
      \frac{x_{p,k}(t)}{\bar{x}_{p,k}} -
      \frac{\bar{x}_{p,k}}{x_{p,k}(t)} \right] \leq 0
  \end{equation*}
  The inequality holds, because $\forall z > 0, \; 2 - z - \frac{1}{z}
  \le 0$ and $\bar{f}_k, \forall k \in K$, $\bar{w}_p, \forall p \in
  P$ are positive.
\end{itemize}
}


\noindent
The above lemma suggests the following mechanism.

\begin{enumerate}

\item For every line pool $k \in K$, solve an instance of the
  single-pool case, using the decentralized mechanism in \cite{BKZ2009},
  obtaining the optimal solution   $(\matrixcol{x}{k}, \matrixcol{\vector{\Lambda}}{k})$.

\item The NOP calculates the cost of each pool and sets the variable
  $\zeta(t)$ to the average pool cost:
  \[
  \zeta(t) = \frac{1}{|K|} \sum_{k \in K}\mv{c}^T \cdot
  \matrixcol{\Lambda}{k}(t).
  \]
  Then, he updates the capacity proportion vector \mv{f} and assigns a
  larger percentage of the total capacity to the most ``expensive''
  line pools, so that their cost decreases. This update is described
  by the following differential equations:
  \begin{equation}
    \label{eq:f_update-2}
    \forall k \in K,\; \dot{f}_k(t) = \max\{0, \mv{c}^T \cdot
    \matrixcol{\Lambda}{k}(t) - \zeta(t)\}.
  \end{equation}
  Note that, at the end, the vector \mv{f} must be normalized, such
  that $\sum_{k \in K} f_k = 1$ (the proportion vector must completely
  divide the infrastructure at all times). This is done by dividing each
  $f_k(t)$ by $\sum_{k \in K}f_k(t)$.

\end{enumerate}

Roughly speaking, the convergence of the above mechanism  for a specific capacity proportion vector \mv{f} is guaranteed by the convergence of the
single-pool algorithm. Nevertheless, for the sake of completeness, we present the convergence proof in the Appendix (Lemma~\ref{lm:convergence}).
When the $|K|$ single-pool instances are solved, the NOP updates the vector \mv{f}, so that the expensive pools get cheaper. The goal is that all pools should have the same cost. When this happens, we know for the optimal solution of both MNET-II and MSC-II $(\bvec{x},\bvec{f})$ and the accompanying Lagrange multipliers, $(\bvec{\Lambda},\bar{\zeta})$, that:
\begin{itemize}
\item Equations \eqref{eq:selfish-KKT-MNET-2} and
  \eqref{eq:selfish-KKT-MSC-2} will match due to the fact that each
  LOP solves the \eqref{eq:muser-II} problem.
\item Equations \eqref{eq:zeta-KKT-MNET-2} (and
  \eqref{eq:zeta-KKT-MSC-2}) will be satisfied due to the updates to
  vector \mv{f} that take place at the end of each phase.
\item Equations \eqref{eq:complementarity1-KKT-MNET-2} (and
  \eqref{eq:complementarity1-KKT-MSC-2}) will be satisfied due to the
  resource price updating scheme chosen.
\item Equations \eqref{eq:complementarity2-KKT-MNET-2} (and
  \eqref{eq:complementarity2-KKT-MSC-2}) will be satisfied, because
  the NOP assures that at the end of each phase.
\end{itemize}

Hence, $(\bvec{x},\bvec{f})$ is the optimal solution of both (KKT-MSC-II) and \eqref{eq:MNET-II}, and thus the proposed mechanism solves \eqref{eq:MSC-diffU}. The next theorem summarizes the preceding discussion.

\begin{theorem}
  The above dynamic scheme of resource pricing, LOPs' bid updating and
  capacity proportion updating assures the monotonic convergence of
  the \eqref{eq:MNET-II} problem to the optimal solution. The
  algorithm may start from any initial state of resource prices, LOPs'
  bids and capacity proportion vector.
\end{theorem}

\section{Experimental Study of the Multiple-Line Pool Cases}
\label{sec:exp-mlp}

In this section we present the experimental results for the
multiple-line pool case where the LOPs have different utility functions
per pool. We have implemented a discrete version of the
decentralized mechanism, whose pseudocode is as follows.
\begin{algorithmic}
  \STATE $f_k(0) = \frac{1}{|K|};$
  \REPEAT
  \STATE $t = t + 1;$
  \FORALL{$k \in K$}
  \STATE Solve an instance of the single-pool case for each line pool $k;$
  \ENDFOR
  \STATE $\mathrm{cost}_k(t) = \mv{c}^T \cdot \matrixcol{\Lambda}{k}(t);$
  \STATE $\zeta = \frac{\sum_{k \in K}\mathrm{cost}_k}{|K|};$
  \FORALL{$k \in K$}
  \STATE $\dot{f}_k(t) = \frac{\max\{0, \mathrm{cost}_k(t) -
    \zeta(t)\}}{\zeta(t)};$
  \STATE $f_k(t) = f_k(t - 1) + 0.1 \cdot \dot{f}_k(t);$
  \ENDFOR
  \STATE $\mathrm{total\_f} = \sum_{k \in K}f_k(t);$
  \FORALL{$k \in K$}
  \STATE $f_k(t) = \frac{f_k(t)}{\mathrm{total\_f}}; \qquad$
  \ENDFOR
  \UNTIL{$\mathrm{equal\_costs}(\mv{\mathrm{cost}}(t))$}
\end{algorithmic}
The algorithm was implemented
in C++ using the GNU g++ compiler (version 4.4) with the second
optimization level (-O2 switch) on.
Experiments were performed on synthetic and real-world data.

Synthetic data consisted of grid graphs having a number of 7 nodes
on the vertical axis and a number of nodes in $[120, 360]$ along
the horizontal axis; i.e., the size of
the grid graphs varied from $7 \times 120$ to $7 \times 360$.
The capacity of each edge was randomly chosen from $[10, 110)$.
Four line pools were defined. In each pool, there were three LOPs,
each one interested in a different line. Those lines had the first
edge $((0,3), (1,3))$ in common. The next edges of each line were
randomly chosen each time.

Real-world data concern parts of the German railway network (mainly
intercity train connections), denoted as R1 (280 nodes and
354 edges) and R2 (296 nodes and 393 edges).
The capacities of the edges were in $[8, 16]$. The total
number of lines varies from 100 up to 1000, depending on the size
of the networks. For each network, we defined four
line pools. The second, third and fourth pool differed from the first
in about 10\% of the lines (the new lines in each pool were randomly
selected from the available lines in each network).

In the experiments we measured the number of iterations needed to
find the correct vector \mv{f} of capacity proportions (we did not
concentrate on the solutions of the single-pool case, used as a
subroutine, since this case was investigated in \cite{BKZ2009}).
\REMOVED{
We started by assuming that the utility function for every LOP in
each pool is given by $U_{p,k}(x_{p,k}) = 100 \cdot \sqrt{x_{p,k}}$.
\begin{table}[!ht]
    \centering
    \subfloat{
    \centering
    \begin{tabular}[t]{c c} %
    \toprule
    $p$ & ~~No.~of updates of \mv{f}\\
    \midrule
    120 & 70\\
    180 & 21\\
    240 & 36\\
    300 & 83\\
    360 & 64\\
    \bottomrule
    \end{tabular}
    }
    \subfloat{
    \ \ \ \ \
    }
    \subfloat{
    \centering
    \begin{tabular}[t]{c c}
    \toprule
    Number of lines~ & ~No.~of updates of \mv{f}\\
    \midrule
    100 & 17\\
    200 & 26\\
    300 & 13\\
    \bottomrule
    \end{tabular}
    }
\caption{Number of updates of \mv{f} for grid graphs of size $7 \times p$ (left) and for R1 (right).}
\label{tab:mp-grid-R1}
\end{table}
\begin{table}[!ht]
    \centering
    \subfloat{
    \centering
    \begin{tabular}[t]{c c}
    \toprule
    Number of lines~ & ~No.~of updates of \mv{f}\\
    \midrule
    100 & 18\\
    200 & 33\\
    300 & 6\\
    400 & 13\\
    500 & 1\\
    \bottomrule
  \end{tabular}
  }
    \subfloat{
    \ \ \ \ \
    }
    \subfloat{
    \centering
    \begin{tabular}[t]{c c}
    \toprule
    Number of lines~ & ~No.~of updates of \mv{f}\\
    \midrule
    100 & 33\\
    200 & 24\\
    300 & 21\\
    400 & 8\\
    500 & 12\\
    600 & 15\\
    700 & 1\\
    800 & 7\\
    900 & 1\\
    1000 & 1\\
    \bottomrule
  \end{tabular}
  }
\caption{Number of updates of \mv{f} for R2 (left) and R3 (right).}
\label{tab:mp-710-gross}
\end{table}
The results for synthetic and real-world data are presented in
Tables~\ref{tab:mp-grid-R1} and \ref{tab:mp-710-gross}.
We observe a small number of necessary updates to the capacity
proportion vector \mv{f}, until the system reaches the optimum. The most
important reason for this is the use of the same utility function for
every pool by the LOPs, because the algorithm starts with the initial
values $f_k = \frac{1}{|K|}$ and the optimal values in this case are
quite close to these initial values.

Consequently, the next step is to use different utility functions
(being also our initial goal) in each pool for each LOP.
}
We investigated the following four scenarios:
\begin{enumerate}
\item[S1:] $U_{p,1}(x_{p,1}) = 10^4 \sqrt{x_{p,1}} \text{ and }
  U_{p,2}(x_{p,2}) = 10^4 \sqrt{x_{p,2}}, \; \forall p \in P$.

\item[S2:] $U_{p,1}(x_{p,1}) = \frac{3}{4} \cdot 10^4 \cdot \sqrt{x_{p,1}}
  \text{ and } U_{p,2}(x_{p,2}) = \frac{4}{5} \cdot 10^4 \cdot
  \sqrt{x_{p,2}}, \; \forall p \in P$.

\item[S3:] $U_{p,1}(x_{p,1}) = 10^4 \cdot \sqrt{x_{p,1}} \text{ and }
  U_{p,2}(x_{p,2}) = \frac{1}{2} \cdot 10^4 \cdot \sqrt{x_{p,2}}, \;
  \forall p \in P$.

\item[S4:] $U_{p,1}(x_{p,1}) = 10^4 \cdot \sqrt{x_{p,1}} \text{ and }
  U_{p,2}(x_{p,2}) = \frac{1}{4} \cdot 10^4 \cdot \sqrt{x_{p,2}}, \;
  \forall p \in P$.
\end{enumerate}
We report on experiments with the R1 network and two line pools
for all four scenarios, and on R2 for scenarios S1 and S2
(similar results hold for the other scenarios).
Table~\ref{tab:mp-dw} shows the results for 100, 200 and 300 lines
per pool for R1, and for 100 to 500 lines per pool for R2.
\begin{table}[!ht]
    \centering
    \subfloat{
    \centering
    \begin{tabular}[t]{c c c c c}
    \toprule
    Number of lines~~ & S1 & S2 & S3 & S4\\
    \midrule
    100 & 9 & 33 & 127 & 178\\
    200 & 12 & 33 & 127 & 178\\
    300 & 19 & 29 & 128 & 178\\
    \bottomrule
  \end{tabular}
  }
  \subfloat{
    \ \ \ \ \
    }
    \subfloat{
    \centering
    \begin{tabular}[t]{c c c}
    \toprule
    Number of lines~~ & S1 & S2\\
    \midrule
    100 & 33 & 52\\
    200 & 26 & 49\\
    300 & 1 & 40\\
    400 & 6 & 34\\
    500 & 1 & 37\\
    \bottomrule
  \end{tabular}
  }
\caption{Number of updates of \mv{f} for different utility functions
and number of lines per pool ($|K| = 2$) for R1 (left) and R2 (right).}
  \label{tab:mp-dw}
\end{table}
For S1 (same utility functions), we observe a small number of necessary updates to the capacity
proportion vector \mv{f}, until the system reaches the optimum. The main
reason for this is the use of the same utility function for
every pool by the LOPs, because the algorithm starts with the initial
values $f_k = \frac{1}{|K|}$ and the optimal values in this case are
quite close to these initial values.
For the other scenarios with different utility functions per pool
(S2, S3, S4), we observe a larger number of the updates required. We also
observe that the more different the utility functions of each LOP
in the two pools are, the larger the number of updates required
to reach the optimum.

Another interesting observation in the case of the different
utility functions per line pool, is that the number of updates of \mv{f} is
almost equal. This is due to the fact that the difference in
utility functions across line pools has a more significant effect
on the required number of updates than the difference in lines
among the pools (in other words, more steps are required to reach
the optimal values due to the different utility functions than due
to the different costs of the line pools).

In conclusion, the number of updates required by our mechanism
to converge (to the optimal values of vector \mv{f}) depends
largely on the exact parameters of the system of
differential equations~\eqref{eq:f_update-2}.

\section{Experimental Study of Disruptions in the Network}
\label{sec:disruptions}

We turn now to a different experimental study. We
assume that the network is currently operating at optimality
and that a few disruptions occur. These disruptions affect the capacity of some edges. This can
be due to technical problems leading to reducing the capacity of
those edges, or to increasing their capacity for a particular period
to handle increased traffic demand (e.g., during holidays, or rush hours)
by ``releasing" more infrastructure.

We examine the behavior of the algorithms for the single and
multiple pool cases in such situations. We investigated three
disruption scenarios:
\begin{enumerate}

\item[D1:] Reducing the capacity of a certain number of edges
  (chosen among the congested ones).

\item[D2:] Increasing the capacity of a certain number of edges
  (chosen among the congested ones).

\item[D3:] Reducing the capacity of a certain number of edges, while
  increasing the capacity of an equal number of a different set of edges 
  (chosen among the congested ones).

\end{enumerate}
%
We start from a known optimal solution to the problem. Then, we
add disruptions to a few edges and apply the algorithm. The relative
and absolute error for the differential equations were set to $0.1$.

These scenarios were tested on grid graphs and on the R1 network
(similar results hold for R2). For the grid graphs,
the lines were chosen randomly, but all of them shared the same first edge.
The number of lines in each pool were 10 and the capacities of the edges
were chosen randomly in $[4, 20]$.

\subsection{Single pool case}
\label{sec:sp}

For this case, we chose randomly, among the congested ones, 4 edges in the
case of grid graphs and 10 edges in the case of R1. Their capacity was
reduced (or increased) by 10\% and 50\%. In the experiments we measured
the number of updates required for finding the optimal values of $\mv{\Lambda}$
(resource prices per-unit-of-frequency).


The number of iterations
required for finding the optimal values of $\mv{\Lambda}$
for grid graphs and R1, when we start from a previous
optimal solution,  is presented in Tables \ref{tab:sp_d-grid} and \ref{tab:sp_d-R1}.
For comparison, the number of the required updates of $\mv{\Lambda}$
when we start from a random initial state is given in Tables \ref{tab:sp_grid}
and \ref{tab:sp_R1}. We observe the significantly
less number of updates required when we start from a previous
optimal solution. This is due to the fact that the disruptions caused
are not very big, and hence the new optimal solution is quite close to
the previous one. There are, however, one or two exceptions; i.e., we observe
in these cases a smaller number of updates when we start from a random initial solution. This
happens, because the algorithms for solving differential equations are
arithmetic methods that depend greatly on the exact parameters
given. This results in a few pathological cases such as these. One can
conclude, though, that in general the use of the previous optimal
solution leads to a smaller number of required updates for \mv{\Lambda}.

\subsection{Multiple pool Case}
\label{sec:mp}

We created two pools for these experiments. In the case of grid graphs,
the lines in each pool were chosen randomly, and in the case of R1
there was a 10\% difference in the lines between the two
pools. In these experiments we measured the number of updates of the
bids of the LOPs (bid vector \mv{w}). In none case there was a need to
update the capacity proportion vector \mv{f}.

The results are shown in Tables \ref{tab:mp_d-grid} and
\ref{tab:mp_d-R1}. One can see that only rarely there is a need to
update the bid vector \mv{w}. Especially for the R1 network, we had to
introduce disruptions of 90\% of the original capacity to get the bid
vector to be updated. Hence, the algorithm reaches the optimal solution
quite fast. The important observation is that, starting from the previous
optimal solution, we avoid the update of the capacity proportion vector
\mv{f}, which is the most expensive operation.

\section{Conclusions}
\label{sec:conclusions}

We have studied a variant of the robust multiple-pool line planning problem defined in
\cite{BKZ2009}, where the LOPs have different utility functions per
pool. We have shown that a dynamic, decentralized mechanism exists for
this problem that eventually converges to the optimal solution.

We have also studied the above mechanism experimentally, showing that
the exact behavior of the algorithm greatly depends on the exact input
parameters; however, the convergence is in general quite fast.

Moreover, we studied the case that disruptions take place in the
network. We have seen that in most cases it is much better to take
advantage of the previous (optimal) solution to bootstrap the
algorithm.
\begin{table}[tbp]
  \centering
  \begin{tabular}{c c r r r}
    \toprule
    Disruption~ & ~~$p$~~ & D1 & D2 & D3\\
    \midrule
    \multirow{5}{*}{10\%} & 120 & 1292 & 340 & 9983\\
    & 180 & 1235 & 395 & 550\\
    & 240 & 317 & 453 & 407\\
    & 300 & 4005 & 556 & 1337\\
    & 360 & 163 & 8484 & 542\\
    \midrule
    \multirow{5}{*}{50\%} & 120 & 403 & 480 & 1022\\
    & 180 & 248 & 1116 & 875\\
    & 240 & 409 & 498 & 533\\
    & 300 & 3966 & 1284 & 1180\\
    & 360 & 751 & 658 & 712\\
    \bottomrule
  \end{tabular}
  \caption{Required number of updates of \mv{\Lambda} for grid graphs
    with sizes $7 \times p$, when the algorithm starts from a previous
    optimal state, until the system reaches the equilibrium
    point after the disruptions under scenarios D1, D2, and D3.}
  \label{tab:sp_d-grid}
\end{table}

\begin{table}[tbp]
  \centering
  \begin{tabular}{c c r r r}
    \toprule
    Disruption~ & ~Number of lines~~ & D1 & D2 & D3\\
    \midrule
    \multirow{3}{*}{10\%} & 100 & 10335 & 90085 & 464\\
    & 200 & 32466 & 2806 & 5033\\
    & 300 & 4171 & 276 & 5208\\
    \midrule
    \multirow{3}{*}{50\%} & 100 & 8409 & 1057 & 1506\\
    & 200 & 1042 & 1109 & 4314\\
    & 300 & 5430 & 974 & 1058\\
    \bottomrule
  \end{tabular}
  \caption{Required number of updates of \mv{\Lambda} for R1, when the
  algorithm starts from a previous optimal state, until
     the system reaches the equilibrium point after the disruptions
     under scenarios D1, D2, and D3.}
  \label{tab:sp_d-R1}
\end{table}
\begin{table}[tbp]
  \centering
  \begin{tabular}[b]{c c r}
    \toprule
    Case of Disruption~ & ~~$p$~~ & ~Number of Updates of \mv{\Lambda}\\
    \midrule
    \multirow{5}{*}{10\%} & 120 & 6701\\
    & 180 & 6643\\
    & 240 & 7835\\
    & 300 & 6813\\
    & 360 & 5854\\
    \midrule
    \multirow{5}{*}{50\%} & 120 & 7381\\
    & 180 & 7246\\
    & 240 & 6468\\
    & 300 & 6197\\
    & 360 & 7617\\
    \bottomrule
  \end{tabular}
  \caption{Number of updates of \mv{\Lambda} for grid graphs of size $7\times p$, when
    the algorithm starts from a random initial state.}
  \label{tab:sp_grid}
\end{table}
\begin{table}[tbp]
  \centering
  \begin{tabular}[b]{c c c}
    \toprule
    Number of lines~ & ~Number of updates of \mv{\Lambda}\\
    \midrule
    100 & 12393\\
    200 & 6641\\
    300 & 7817\\
    \bottomrule
  \end{tabular}
  \caption{Number of updates of \mv{\Lambda} for R1, when
    the algorithm starts from a random initial state.}
  \label{tab:sp_R1}
\end{table}
\begin{table}[tbp]
  \centering
  \begin{tabular}{c c r r r}
    \toprule
    Disruptions~ & ~~$p$~~ & D1 & D2 & D3\\
    \midrule
    \multirow{5}{*}{10\%} & 120 & 0 & 0 & 0\\
    & 180 & 0 & 0 & 0\\
    & 240 & 0 & 0 & 0\\
    & 300 & 0 & 0 & 0\\
    & 360 & 0 & 0 & 0\\
    \midrule
    \multirow{5}{*}{50\%} & 120 & 0 & 2 & 1\\
    & 180 & 0 & 2 & 0\\
    & 240 & 0 & 0 & 0\\
    & 300 & 0 & 1 & 2\\
    & 360 & 0 & 2 & 2\\
    \bottomrule
  \end{tabular}
  \caption{Required number of updates of \mv{w} for grid graphs of
  sizes $7\times p$ for scenarios D1, D2, D3, when the
  algorithm starts from a previous optimal state, so that the system returns
  to an equilibrium point after a disruption.}
  \label{tab:mp_d-grid}
\end{table}

\begin{table}[tbp]
  \centering
  \begin{tabular}{c c r r r}
    \toprule
    Disruption~ & ~Number of lines~ & D1 & D2 & D3\\
    \midrule
    \multirow{3}{*}{10\%} & 100 & 0 & 0 & 0\\
    & 200 & 0 & 0 & 0\\
    & 300 & 0 & 0 & 0\\
    \midrule
    \multirow{3}{*}{50\%} & 100 & 0 & 0 & 0\\
    & 200 & 0 & 0 & 0\\
    & 300 & 0 & 0 & 0\\
    \midrule
    \multirow{3}{*}{90\%} & 100 & 0 & 3 & 0\\
    & 200 & 0 & 2 & 2\\
    & 300 & 0 & 0 & 0\\
    \bottomrule
  \end{tabular}
  \caption{Required number of updates of \mv{w} for R1 for
    scenarios D1, D2, D3, when the algorithm starts from a
    previous optimal state, so that the system returns to an
    equilibrium point after a disruption.}
    \label{tab:mp_d-R1}
\end{table}

\section*{APPENDIX}

\REMOVED{
\subsection*{A.1 Pseudocode of Algorithm}

The pseudocode of the discrete version of our decentralized mechanism is as follows:

\begin{algorithmic}
  \STATE $f_k(0) = \frac{1}{|K|};$
  \REPEAT
  \STATE $t = t + 1;$
  \FORALL{$k \in K$}
  \STATE Solve an instance of the single-pool case for each line pool $k;$
  \ENDFOR
  \STATE $\mathrm{cost}_k(t) = \mv{c}^T \cdot \matrixcol{\Lambda}{k}(t);$
  \STATE $\zeta = \frac{\sum_{k \in K}\mathrm{cost}_k}{|K|};$
  \FORALL{$k \in K$}
  \STATE $\dot{f}_k(t) = \frac{\max\{0, \mathrm{cost}_k(t) -
    \zeta(t)\}}{\zeta(t)};$
  \STATE $f_k(t) = f_k(t - 1) + 0.1 \cdot \dot{f}_k(t);$
  \ENDFOR
  \STATE $\mathrm{total\_f} = \sum_{k \in K}f_k(t);$
  \FORALL{$k \in K$}
  \STATE $f_k(t) = \frac{f_k(t)}{\mathrm{total\_f}}; \qquad$
  \ENDFOR
  \UNTIL{$\mathrm{equal\_costs}(\mv{\mathrm{cost}}(t))$}
\end{algorithmic}

\subsection*{A.2 Omitted Proofs}
}

\textbf{Proof of Lemma~\ref{lm:main}}
\begin{proof}
  The Lagrangian function of the problem is:
  \begin{equation*}
    L(\mv{x}, \mv{f}, \mv{\Lambda}, \zeta) = \\
    \sum_{p\in P}\sum_{k \in K} \left[ U_{p,k}(x_{p,k}) - x_{p,k}
      \cdot \mu_{p,k}(\mv{\Lambda}) \right] + \sum_{k\in K} f_k \cdot
    \left[\mv{c}^T\matrixcol{\Lambda}{k} - \zeta \right] + \zeta.
  \end{equation*}
  The dual problem of \eqref{eq:MSC-diffU} is the following:
  \begin{equation*}
    \begin{aligned}
      \max \quad & D(\mv{\Lambda}, \zeta)\\
      s.t. \quad & \Lambda_{\ell,k} \ge 0 \quad \forall \ell \in L, \,
      \forall k \in K\\
      & \zeta \ge 0,
    \end{aligned}
  \end{equation*}
  where
  \begin{align*}
    & D(\mv{\Lambda}, \zeta)  =\\
    & = \max\left\{L(\mv{x}, \mv{f}, \mv{\Lambda}, \zeta) \colon
      \mv{x}, \mv{f} \ge \mv{0} \right\} \\
    & = \max_{\mv{x}, \mv{f} \ge \mv{0}} \left\{ \sum_{p \in P}\sum_{k
        \in K}\left[ U_{p,k}(x_{p,k}) - x_{p,k}\sum_{l \in
          L}\Lambda_{\ell, k}R_{\ell, p}(k) \right] + \sum_{k \in
        K}f_k \left[ \sum_{\ell \in L}\Lambda_{\ell,k}c_{\ell} - \zeta
      \right] +
      \zeta \right\}\\
    & = \max_{\mv{x} \ge \mv{0}} \left\{ \sum_{p \in P}\sum_{k \in
        K}\left[ U_p(x_p) - x_{p,k}\mu_{p,k}(\mv{\Lambda}) \right]
    \right\} + \max_{\mv{f} \ge \mv{0}} \left\{ \sum_{k \in K}f_k
      \left[ \sum_{\ell \in L}\Lambda_{\ell, k}c_{\ell} - \zeta
      \right] \right\} + \zeta.
  \end{align*}

  For the part
  \begin{equation*}
    P_1(\mv{\Lambda}) = \max_{\mv{x} \ge \mv{0}} \left\{ \sum_{p \in
        P} \left[ U_p(x_p) - \sum_{k \in K}x_{p,k} \cdot
        \mu_{p,k}(\mv{\Lambda}) \right] \right\}
  \end{equation*}
  we have, that it is a simple maximization problem the solution of
  which depends exclusively on the optimal vector of resource prices
  \bvec{\Lambda}.

  For the part
  \begin{align*}
    P_2(\mv{\Lambda}, \zeta) & = \max_{\mv{f} \ge \mv{0}} \left\{
      \sum_{k \in K} f_k \cdot \left[ \sum_{\ell \in
          L}\Lambda_{\ell,k} \cdot c_{\ell} - \zeta \right] \right\}
    + \zeta\\
    & = \max_{\mv{f} \ge \mv{0}} \left\{ \sum_{k \in K} \cdot \left(
        \mv{c}^T \cdot \matrixcol{\Lambda}{k} \right) + \zeta \cdot
      \left( 1 - \sum_{k \in K}f_k \right) \right\}
  \end{align*}
  we have, that at optimality $\zeta \cdot \left( 1 - \sum_{k \in
      K}f_k \right)$ equals to zero. So, we get:
  \begin{align*}
    P_2(\mv{\Lambda}, \zeta) & = \max_{\mv{1}^T\mv{f} = 1 \colon
      \mv{f} \ge \mv{0}} \left\{ \sum_{k \in K}f_k \cdot \left(
        \mv{c}^T \cdot \matrixcol{\Lambda}{k} \right) \right\}\\
    & = \max_{k \in K} \left\{ \mv{c}^T \cdot \matrixcol{\Lambda}{k}
    \right\} \\
    & = \min \left\{ \mv{z} \colon \mv{z} \cdot \mv{1}^T \ge \mv{c}^T \cdot
      \mv{\Lambda} \right\}.
  \end{align*}
  That is, $P_2(\mv{\Lambda}, \zeta)$ calculates the maximum cost
  among the pools. This, however, depends only on the optimal vector
  of resource prices \bvec{\Lambda}.
\qed
\end{proof}
The convergence of our decentralized mechanism is shown by the next lemma.
\begin{lemma}
\label{lm:convergence}
For any choice of \emph{fixed} bid vector $\bvec{w} = (w_p)_{p\in P}$ offered by the LOPs,
and any \emph{fixed} vector of proportions $\bvec{f} = (f_k)_{k\in K}$ determined by the NOP,
the resource price updating scheme makes the resource prices converge to the corresponding
optimal vector $\bvec{\La}$ (for these particular given bids and proportions).
\end{lemma}

\begin{proof}
We use the Lyapunov function
\begin{equation*}
  V(\mm{\Lambda}(t)) = \frac{1}{2} \cdot \left( \mm{\Lambda}(t) -
    \bvec{\Lambda} \right)^T \cdot \left( \mm{\Lambda}(t) -
    \bvec{\Lambda} \right)
\end{equation*}
and showing that $\frac{dV(\mv{\Lambda}(t))}{dt} \le 0$. Indeed,
\begin{align*}
  \frac{dV(\mm{\Lambda}(t))}{dt} & =\\
  & = \sum_{\ell\in L}\sum_{k\in K} \left(
    \Lambda_{\ell,k}(t)-\bar{\Lambda}_{\ell,k} \right) \cdot
  \dot{\Lambda}(t) \\
  & = \sum_{\ell\in L}\sum_{k\in K} \left(
    \Lambda_{\ell,k}(t)-\bar{\Lambda}_{\ell,k} \right) \cdot \left[
    \max\left\{0, y_{\ell,k}(t) - c_{\ell}\bar{f}_k \right\} \cdot
    \Ind{\Lambda_{\ell,k}(t)=0} \right. \\
  & \left.\hspace*{5cm} + \left(y_{\ell,k}(t) -
      c_{\ell}\bar{f}_k\right) \cdot \Ind{\Lambda_{\ell,k}(t)>0}
  \right]\\
  & \leq \sum_{\ell\in L}\sum_{k\in K} \left( \Lambda_{\ell,k}(t) -
    \bar{\Lambda}_{\ell,k} \right) \cdot \left[ y_{\ell,k}(t) -
    c_{\ell} \bar{f}_k \right] \\
  & = \sum_{\ell\in L}\sum_{k\in K} \left( \Lambda_{\ell,k}(t) -
    \bar{\Lambda}_{\ell,k} \right) \cdot \left[ y_{\ell,k}(t) -
    \bar{y}_{\ell,k} + \bar{y}_{\ell,k} - c_{\ell} \bar{f}_k
  \right]\\
  & \leq \sum_{\ell\in L}\sum_{k\in K} \left( \Lambda_{\ell,k}(t) -
    \bar{\Lambda}_{\ell,k} \right) \cdot \left[ y_{\ell,k}(t) -
    \bar{y}_{\ell,k} \right] \\
  & = \sum_{\ell\in L}\sum_{k\in K} \left( \Lambda_{\ell,k}(t) -
    \bar{\Lambda}_{\ell,k} \right) \cdot \sum_{p\in P} R_{\ell,p}(k)
  \cdot \left[ x_{p,k}(t) - \bar{x}_{p,k} \right]\\
  & = \sum_{p\in P}\sum_{k\in K} \left[ x_{p,k}(t) - \bar{x}_{p,k}
  \right] \cdot \left( \mu_{p,k}(t) - \bar{\mu}_{p,k} \right)\\
  & = \sum_{p\in P}\sum_{k\in K} \left[ x_{p,k}(t) - \bar{x}_{p,k}
  \right] \cdot \left( \frac{\bar{w}_{p,k}}{x_{p,k}(t)} -
    \frac{\bar{w}_{p,k}}{\bar{x}_{p,k}} \right)\\
  & = \sum_{p\in P}\sum_{k\in K}\bar{w}_{p,k} \cdot \left[2 -
    \frac{x_{p,k}(t)}{\bar{x}_{p,k}} -
    \frac{\bar{x}_{p,k}}{x_{p,k}(t)} \right]\\
  & \leq 0
\end{align*}
The first inequality
  \begin{align*}
    & \sum_{\ell\in L}\sum_{k\in K} \left(
      \Lambda_{\ell,k}(t)-\bar{\Lambda}_{\ell,k} \right) \cdot \left[
      \max\left\{0, y_{\ell,k}(t) - c_{\ell}\bar{f}_k \right\} \cdot
      \Ind{\Lambda_{\ell,k}(t)=0} \right. \\
    & \left.\hspace*{5cm} + \left(y_{\ell,k}(t) -
        c_{\ell}\bar{f}_k\right) \cdot \Ind{\Lambda_{\ell,k}(t)>0}
    \right]\\
    & \leq \sum_{\ell\in L}\sum_{k\in K} \left( \Lambda_{\ell,k}(t) -
      \bar{\Lambda}_{\ell,k} \right) \cdot \left[ y_{\ell,k}(t) -
      c_{\ell}
      \bar{f}_k \right]\\
  \end{align*}
  clearly holds $ \forall (\ell, k) \colon \Lambda_{\ell,k}(t)
  > 0$, but it also holds when $ \forall (\ell, k) \colon \Lambda_{\ell,k}(t) =
  0$. The last part holds, because it will either be $y_{\ell,k}(t) -
  c_{\ell}\bar{f}_{k} \ge 0$ (then it is obvious) or
  $y_{\ell,k}(t) - c_{\ell}\bar{f}_k < 0$. In this case:
  \begin{equation*}
    \left( \Lambda_{\ell,k}(t) - \bar{\Lambda}_{\ell,k} \right) \cdot
    \max\left\{0, y_{\ell,k}(t) - c_{\ell}\bar{f}_k \right\} = -
    \bar{\Lambda}_{\ell,k} \cdot 0 < -\bar{\Lambda}_{\ell,k} \cdot
    \left[ y_{\ell,k}(t) - c_{\ell}\bar{f}_{k} \right].
  \end{equation*}
The second inequality
  \begin{align*}
    & \sum_{\ell\in L}\sum_{k\in K} \left( \Lambda_{\ell,k}(t) -
      \bar{\Lambda}_{\ell,k} \right) \cdot \left[ y_{\ell,k}(t) -
      \bar{y}_{\ell,k} + \bar{y}_{\ell,k} - c_{\ell} \bar{f}_k
    \right]\\
    & \leq \sum_{\ell\in L}\sum_{k\in K} \left( \Lambda_{\ell,k}(t) -
      \bar{\Lambda}_{\ell,k} \right) \cdot \left[ y_{\ell,k}(t) -
      \bar{y}_{\ell,k} \right]\\
  \end{align*}
  holds because for the optimal matrix of
  resource prices \bvec{\Lambda} (for given \bvec{w} and \bvec{f}) it
  holds that
  \begin{equation*}
    \sum_{\ell \in L}\sum_{k \in K} \bar{\Lambda}_{\ell,k}
    (\bar{y}_{\ell,k} - c_{\ell}\bar{f}_k) = 0
  \end{equation*}
  (due to equation~\eqref{eq:complementarity1-KKT-MSC-2}).
The third inequality
  \begin{equation*}
    \sum_{p\in P}\sum_{k\in K}\bar{f}_k\bar{w}_p \cdot \left[2 -
      \frac{x_{p,k}(t)}{\bar{x}_{p,k}} -
      \frac{\bar{x}_{p,k}}{x_{p,k}(t)} \right] \leq 0
  \end{equation*}
  holds because $\forall z > 0, \; 2 - z - \frac{1}{z}
  \le 0$ and $\bar{f}_k, \forall k \in K$, $\bar{w}_p, \forall p \in
  P$ are positive.
\qed
\end{proof}

\end{document}